\documentclass[11pt]{article}

\usepackage[ruled]{algorithm2e} % For algorithms

\usepackage{amsmath,amsfonts,amsthm}
\usepackage{xspace}
\usepackage[dvipsnames]{xcolor}
\usepackage{url}
\usepackage{authblk}
\usepackage{caption}
\usepackage{subcaption}
\usepackage{tikz}
\usepackage{paralist}

\usepackage[margin=1in]{geometry}
\usepackage{txfonts}
\usepackage[T1]{fontenc}
\usepackage{libertine}
\usepackage[
pagebackref,
colorlinks=true,
urlcolor=blue,
linkcolor=blue,
citecolor=OliveGreen,
]{hyperref}
\usepackage[nameinlink]{cleveref}
\usepackage{enumitem}

\usepackage[suppress]{color-edits}
\addauthor{rdk}{blue}
\addauthor{ma}{red}

\newtheorem{theorem}{Theorem}[section]
\newtheorem{lemma}[theorem]{Lemma}

\theoremstyle{definition}
\newtheorem{definition}{Definition}[section]
\newtheorem{example}{Example}[section]

\newcommand{\F}{\mathcal{F}}
\newcommand{\argmax}{\text{argmax}}
\newcommand{\expect}{\mathbb{E}}

\newcommand{\eps}{\varepsilon}
\newcommand{\supp}{\mathrm{supp}}
\newcommand{\Support}{\mathcal{S}}
\DeclareMathOperator{\ALG}{ALG}

\newcommand{\ind}{\mathbb{I}}
\DeclareMathOperator{\OPT}{OPT}
\newcommand{\OptOn}{\textnormal{\textsc{OPTOn}}}
\newcommand{\OptOff}{\textnormal{\textsc{OPTOff}}}
\newcommand{\OptOffIIF}{\textnormal{\textsc{OPTOffIIF}}}
\newcommand{\OptOnIIF}{\textnormal{\textsc{OPTOnIIF}}}
\newcommand{\OnTIF}{\textnormal{\textsc{OnTIF}}}
\newcommand{\OptOnMHIIF}{\textnormal{\textsc{OPTOnMhIIF}}}
\newcommand{\I}{\mathcal{I}}
\renewcommand{\wp}{\mathrm{w.p.}\ \ }
\newcommand{\bigO}{\mathcal{O}}
\newcommand{\st}{\text{s.t.}}

\title{Individual Fairness in Prophet Inequalities}
\author[1]{Makis Arsenis}
\author[1]{Robert Kleinberg}
\affil[1]{Cornell University, Ithaca, NY\protect\\
{\footnotesize \texttt{\{marsenis,rdk\}@cs.cornell.edu}}}
\date{}

\begin{document}

% Title page for title and abstract only.
\begin{titlepage}
\maketitle

\begin{abstract}
Prophet inequalities are performance guarantees for online
algorithms (a.k.a.~stopping rules) solving the following
``hiring problem'': a decision
maker sequentially inspects candidates whose values are 
% random
% variables $X_i$ distributed according
% to known distributions $\F_i$ in order from $1$ to
% $n$ 
independent random numbers 
and is asked to hire at most one candidate
by selecting it before inspecting the 
values of future candidates in the sequence.
% for each candidate $i$ right
% after inspecting their value $X_i$.
% The online nature of the
% problem makes it difficult to solve optimally (you cannot go
% back to a candidate after you rejected them if you realize
% later they were the best in hindsight). 
A classic
result in optimal stopping theory asserts that there exist
stopping rules guaranteeing that the decision maker will
hire a candidate whose expected value is at least
\emph{half} as good as the expected value of the candidate
hired by a ``prophet'', i.e.~one who has simultaneous access
to the realizations of all candidates' values.

Such stopping rules may indeed have provably good
performance but might treat individual candidates unfairly
in a number of different ways. In this work we identify two
types of individual fairness that might be desirable in
optimal stopping problems. We call them
\emph{identity-independent fairness} (IIF) and
\emph{time-independent fairness} (TIF) and give precise
definitions in the context of the hiring problem. We give
polynomial-time algorithms for finding the optimal IIF/TIF
stopping rules for a given instance with discrete support
and we manage to recover a prophet inequality with
factor $1/2$ when the
decision maker's stopping rule is required to satisfy both
fairness properties while the prophet is unconstrained. We
also explore worst-case ratios between optimal selection rules
in the presence vs.~absence of individual fairness constraints, 
in both the online and offline settings. 
% whose constraints differ on two dimensions: offline
% vs.~online and non-fair vs.~IIF/TIF. 
We prove
an impossibility result showing that there is no prophet
inequality with a non-zero factor for either IIF or TIF
stopping rules when we further constrain the decision maker
to make a hire with probability~$1$. We finally consider
a setting in which the decision maker doesn't know the
distributions of candidates' values but has access to
a bounded number of independent samples from each distribution.
We provide constant-competitive algorithms that 
satisfy both TIF and IIF, using one sample from each
distribution in the offline setting and two samples
from each distribution in the online setting.
\end{abstract}

\end{titlepage}

\section{Introduction}
\label{sec:intro}

Optimal stopping problems, which model the problem of
deciding when to select an element of a random sequence
amid uncertainty about the future elements to be sampled,
are ubiquitous in the search and matching problems 
that underlie the design of dynamic markets.
In optimal stopping, as in many other decision
problems, there is a tension between optimality 
--- i.e., maximizing the expected value of the chosen
alternative --- and fairness, i.e.~ensuring equal 
opportunity for all alternatives. In contrast to
prior work that investigated optimal stopping 
in the context of {\em group fairness} 
constraints \cite{correa2021fairness},
in this work we initiate the study of how
{\em individual fairness constraints} influence
the performance of stopping rules.

Our approach to this question is inspired by 
the literature on {\em prophet inequalities}, 
which seeks to understand the performance of 
stopping rules by proving theorems that identify
worst-case bounds on the ratio between the 
expected performance of different types of
selection rules --- for example, comparing
the optimal stopping rule with an offline 
selection rule that always selects the 
best alternative in hindsight. In the same 
spirit, we seek worst-case multiplicative 
comparisons between optimal stopping rules
(or offline selection rules) with and 
without various fairness constraints. 

In this paper we work with an abstract 
formulation of individual fairness in optimal 
stopping that is as 
simple and free of application-specific details
as possible. However, to motivate the
fairness constraints we will consider, 
it will be helpful to contemplate some motivating 
applications in which the potential for
stopping rules to yield unfair outcomes
is evident. 

\begin{example} \label{ex:firm-hiring}
Consider a firm 
interviewing a sequence of candidates
for a job opening. Assume that the firm
may hire at most one candidate, and that
the decision whether or not to hire a
candidate must be made immediately after
their interview, without waiting to judge
the quality of candidates scheduled to be 
interviewed later in the hiring season. 
If we assume that the firm wishes to
maximize the expected quality of the
candidate hired, and that the candidates'
qualities are 
independent and identically 
distributed\footnote{The i.i.d.~assumption
is for the sake of this example. In general
we will be considering independent but non-identical 
distributions in this work.}
random variables, then an optimal
stopping rule would calculate a decreasing
sequence of thresholds and hire the 
first candidate whose quality exceeds
the corresponding threshold. Thus, although
the candidates are {\em a priori} identical,
the optimal stopping rule discriminates according
to arrival time: a candidate of low
quality has no chance of being hired
if they interview early (when the threshold
will be above their quality level) but 
stands a chance of being hired if they
interview late. On the other hand, the 
opposite type of unfairness exists for 
candidates of high quality: they have a
high probability of being hired if they
interview early, and a lower probability
of being hired if they interview late
because of the possibility that an 
earlier candidate has already been hired.
Replacing the optimal stopping rule with
a threshold stopping rule --- i.e., one
which uses a constant threshold over time
and hires the first candidate whose quality
exceeds this threshold --- eliminates the
first type of unfairness but not the second.
\end{example}

In the preceding example, the selection
rule treated individuals differently due
to differences in their arrival time. 
Another potential type of unfairness 
occurs when an individual's probability
of being selected (conditional on their
value) depends on the individual's identity.
This type of unfairness can arise even in
offline selection problems, where a 
decision maker is choosing one of $n$
elements and observes the value (or
quality) of each element before making
a choice. 

\begin{example} \label{ex:rideshare}
Consider a ride
sharing platform in a city where demand
exceeds supply: in each time interval, 
the platform receives a set of requests 
and must decide which ones will receive 
service. A seemingly fair selection rule
is to choose the request of maximum value,
breaking ties at random. To illustrate the
potential for unfairness, suppose
there is only one driver and two users,
Odysseus and Penelope. Odysseus is a 
frequent traveler who requests a ride
every day, whereas Penelope mostly stays
home and only needs a ride once every 
$m$ days for some $m > 1$. Let us model 
the value of selecting a user to be 
1 if they requested a ride and 0 otherwise.
If the platform uses random tie-breaking
then half of Penelope's requests
are denied whereas the
fraction of Odysseus' requests that
are denied is only $\frac{1}{2m}$. 
A fairer rule would select Penelope
with probability $\frac{m}{m+1}$ when
she makes a request, and otherwise it 
would select Odysseus. Under this rule,
both users would find that $\frac{1}{m+1}$
of their requests are denied and the rest
are served. 
\end{example}

The foregoing two examples motivate two
different notions of individual fairness.
We formalize these two properties 
in \Cref{sec:fairness} and label
them as {\em time-independent fairness}
(TIF) and {\em identity-independent fairness}
(IIF). Depending on the application, one may 
wish for selection rules to satisfy one of
these properties, or the other, or both.
A few natural questions arise about such
fairness-constrained selection rules.
\begin{compactenum}
    \item {\bf Can one efficiently optimize
    over fairness-constrained stopping
    rules?} Absent fairness constraints,
    one can generally compute optimal stopping
    rules using dynamic programming, but it 
    appears difficult to incorporate 
    fairness constraints into the dynamic
    programming formulation since they
    impose dependencies
    among the decisions made at different 
    times. We show in \Cref{sec:characterization}
    that optimal fairness-constrained
    stopping rules can be computed by
    solving a polynomial-sized linear 
    program.
    \item {\bf How well can fairness-constrained
    selection rules approximate optimal 
    selection rules, in the worst case?}
    There are many versions of this question,
    depending whether the fairness constraint
    is TIF or IIF (or both), and whether the
    fairness-constrained selection rule 
    or the optimal selection rule (or both) 
    is allowed to make decisions offline
    (i.e., observe the full sequence of 
    values before making its selection). 
    In \Cref{sec:competitive} we answer
    almost all versions of this question;
    our results are summarized in
    \Cref{fig:squares}. In \Cref{sec:impossibility}
    we investigate how the answers to these
    questions differ when one imposes an
    additional constraint that the decision
    maker must select one of the $n$ alternatives.
    In the traditional setting of prophet inequalities
    one can assume this property without loss of
    generality: if no hire is made before the last
    interview, a firm loses nothing by hiring the 
    final candidate. In contrast, we show that requiring 
    the firm to hire a candidate, while also imposing
    either the TIF or IIF fairness constraint, can make an enormous 
    qualitative difference: the expected value of
    the candidate selected by the optimal TIF (or IIF) 
    stopping rule
    may differ from the expected value of the best candidate
    by an unbounded factor.
    \item {\bf In cases when the precise input
    distribution is not known, can 
    approximately optimal 
    fairness-constrained selection rules
    be learned from data?} Since the fairness
    criteria in this paper are distribution-dependent,
    one might anticipate that there is no way 
    for a decision maker to ensure that their 
    selection rule is fair without
    knowing the input distribution. In 
    \Cref{sec:double-sample} we show that
    this intuition is false: if the decision 
    maker has access to a single sample from 
    each distribution, that is enough to 
    implement a constant-competitive offline
    selection rule that satisfies both TIF and IIF.
    With access to {two} independent samples from
    each distribution, the decision maker can implement
    a constant-competitive online stopping rule that
    satisfies both TIF and IIF. The question of 
    whether fairness-constrained stopping rules
    can be constant-competitive with just one
    sample from each distribution, rather than two,
    is an enticing open question that we leave for
    future work.
\end{compactenum}

\subsection{Related Work}

Our work fits the line of research on optimal stopping
problems initiated by \cite{KrengelSucheston} and followed
by others including \cite{SamuelCahn} and more recent works
highlighting the connection of stopping problems with
Mechanism Design initiated by \cite{Hajiaghayi07}. Since then,
research in the area flourished, and a
number of variants of optimal stopping problems have been
studied such as ones that allow hiring multiple people
(e.g.~\cite{Hajiaghayi07,KleinbergWeinberg}), giving the
interviewer the ability to choose the order of inspection
(e.g.~\cite{yan10}) or letting nature permute candidates
uniformly at random before being interviewed, a setting
known as a prophet secretary problem
(\cite{Esfandiari2007}). Another setting, closely related to
the one we consider in \Cref{sec:double-sample} is when the
decision maker doesn't have access to the proability
distributions of the random variables but instead on one (or
a few) samples from each (e.g.~\cite{Azar2014,rubinstein20}).

Among the prophet inequalities literature, fairness has been
very recently considered in \cite{correa2021fairness} but
unlike our work which deals with individual fairness (from
the perspective of each candidate),
\cite{correa2021fairness} deal exclusively with a
form of group-fairness constraint for the secretary problem
and prophet inequality settings.
In their prophet model, individuals are partitioned into disjoint
groups, and fairness constraints are expressed as lower
bounds on the expected number of members selected from each
group. They prove a tight competitive ratio of $1/2$ between
a fair, online algorithm and a fair, offlline one, akin to
some of our results. Additionally they compare fair, online
algorithms to fair, offline ones in some restricted settings
(e.g.~group hiring priorities being proportional to group size,
distributions being i.i.d. and more) and provide tight
competitive ratios there as well.

Somewhat closer to our setting, a similar notion to our TIF
fairness has been studied in the context of the secretary
problem under the name of ``Incentive Compatibility'' in
\cite{Buchbinder}. This is the property of a stopping rule
for the secretary problem which requires the probability of
selecting a candidate at the $t$-th step to be equal for
all $t \in [n]$, hence the name incentive compatible since
any candidate $i$ has no incentive to arrive at any
different time. Interpreted within our framework, the
incentive compatibility of \cite{Buchbinder} can be
thought of as an ex-ante kind of time-independent
fairness, whereas our TIF is an ex-interim kind since our
definition conditions on the sampled value of the $i$-th
candidate when considering their hiring probability.

In a broader sense, our work adds to a line of research
exploring fairness in dynamic settings such as
dynamic resource allocation problems
(e.g.~\cite{radEC21,lien14,Sinclair2020}) or dynamic fair
division (e.g.~\cite{Walsh11,Kash13}), settings in which
agents arrive dynamically and some allocation needs to
be computed that is both maximizing some performance metric
but also satisfying certain fairness criteria.

\section{Preliminaries}
\label{sec:prelim}

A \emph{prophet inequality} is a competitive ratio guarantee
for the following online optimization problem which we will
refer to as the \emph{(online) hiring problem}: a sequence
of $n$ \emph{candidates} numbered $1$ through $n$ are being
interviewed in a pre-defined order dictated by a fixed
permutation\footnote{Our description of the problem differs
slightly from what is standard in the literature.  Usually
the candidates are numbered according to their arrival order
and there is no distinction between time-steps and the
identities of the candidates.} $\pi \in S_n$ (we use $S_n$
to denote the symmetric group of order $n$ containing all
permutations on $n$ elements). This means that on the $t$-th
iteration, candidate $\pi(t)$ arrives to be interviewed.
Each candidate $i$ has an associated non-negative
\emph{value} $X_i$ which is a random variable distributed
according to a distribution $\F_i$. The \maedit{tuple of all
distributions $(\F_i)_{i = 1}^n$}, the value of $n$ and the
permutation $\pi$ are all revealed to the interviewer before
any candidates arrive.  We will denote an instance like that
as $\I = ((\F_i)_{i = 1}^n, \pi)$.\footnote{Sometimes, we
need to consider an instance without an attached arrival
ordering in which case the permutation will simply be
omitted.} Upon arrival of candidate $i = \pi(t)$ at time
step $t$, their value $X_i$ is revealed to the interviewer
who now has to make a decision on whether to hire them or
not.  If they decide to hire candidate $i$, the process
terminates, otherwise they have to reject them irrevocably
and proceed to interview the next candidate.

The interviewer is allowed to follow any algorithm (even
randomized) which at each step $t$ makes a decision based on
the knowledge revealed to them so far (i.e.~the values
$X_{\pi(1)}, \ldots, X_{\pi(t)}$, the knowledge of priors
$\F_i$, the value of $n$ and the permutation $\pi$). We will
use the terms \emph{online algorithm} and \emph{stopping
rule} interchangeably to refer to an algorithm like that. In
all algorithms we consider in this paper, we assume they are
parameterized by $(\F_i)_{i = 1}^n, n$ and $\pi$ (if
present), meaning that those values are hard-coded in the
algorithm and the \emph{input} of the algorithm consists of
the values $X_i$ revealed at each step. When an algorithm
is given a name (e.g.~$\ALG$)\footnote{Sometimes, a
particular permutation $\pi$ will be denoted as a
superscript to indicate that this is an algorithm intended
to run on a particular arrival ordering.} in this or other
settings that follow, we assume its name also serves as a
random variable denoting the value of the candidate hired by
this algorithm (and zero if no person was hired). Under this
convention, $\expect[\ALG]$ refers to the expected value of
the person hired by $\ALG$ where expectation is over the
randomness of the distributions $\F_i$ as well as any
randomness used by the algorithm.

The goal of the interviewer is to follow an algorithm for
the given instance of the problem that maximizes the
expected value of the candidate hired. To evaluate the
performance of an algorithm followed by the interviewer we
compare it to the expected performance of an offline
analogue, an algorithm which has simultaneous access to the
realizations of all the values $X_i$ before being asked to
make a hiring decision. This offline version of the hiring
problem will be common in what follows and we refer to it as
the \emph{offline setting} and refer to an algorithm
followed by an offline interviewer simply as
\emph{algorithm}.  An optimal offline algorithm,
colloquially called a \emph{prophet}, is able to just pick
one of the best candidates $i^* \in \argmax_{i = 1}^n X_i$
regardless of the arrival permutation $\pi$.

In this setting, a prophet inequality is a statement of the
form ``for all instances $\I = ((\F_i)_{i=1}^n, \pi)$
there exists a stopping rule $\ALG$ for which
$\expect[ \ALG ] \ge c \cdot \expect[ X_{i^*}]$ where
$i^*$ is the index chosen by the prophet'', proven for some
positive constant $c < 1$. The \maedit{supremum over all
constants $c$ for which a prophet inequality can be proven}
is called the \emph{competitive ratio} \maedit{of the hiring
problem}.
\macomment{Maybe this is a good spot to introduce the
generic competitive ratio notation we used on the slides.}

The classic result of \cite{KrengelSucheston} asserts a
prophet inequality for $c = 1/2$. In other words, in any
instance of the hiring problem, and regardless of the
arrival ordering, there is a stopping rule $\ALG$ which on
expectation performs at least half as good as the prophet:
\[ \expect[\ALG] \ge \frac12 \cdot \expect[X_{i^*}] \]

Further, the constant of $c = 1/2$ is tight as the following
example shows, \maedit{which means that the competitive
ratio of the classic hiring problem is exactly $1/2$}:

\paragraph{Tightness of the prophet inequality} Consider the
following instance of two random variables parameterize by a
small constant $\eps > 0$ arriving according to the order
they are numbered ($\pi$ is the identity permutation):

\[
X_1 = 1,\quad
X_2 =
\begin{cases}
\frac{1}{\eps}, &\wp \eps\\
0, & \wp 1 - \eps
\end{cases}
\]

On expectation, the prophet receives a value of
$\expect[\max_{i = 1}^n X_i] = 1 / \eps \cdot \eps + 1 \cdot
(1 - \eps) = 2 - \eps$.  On the other hand, a
(deterministic) interviewer who is presented with $X_1$
first has no additional information to make a decision so
they must either commit to hiring $X_1$, or reject them and
be presented with the riskier choice of candidate $2$, at
which point they would have no choice but to hire.  In
either case, the interviewer is getting either $\expect[X_1]
= 1$ or $\expect[X_2] = 1$ for $\expect[\ALG] = 1$. Letting
$\eps \to 0$ gives a multiplicative gap approaching $1/2$
between the expected performance of an interviewer and that
of a prophet.\footnote{We presented the gap in the case of
an interviewer who follows a deterministic stopping rule but
it's not hard to verify that the same instance gives a $1/2$
gap even for a randomized interviewer.}

Part of the reason why prophet inequalities have become so
important in Mechanism Design is due to the fact that there
actually exist very simple-to-describe threshold stopping
rules that achieve the optimal competitive ratio of $1/2$.
A \emph{threshold stopping rule} is an online algorithm for
the hiring problem which pre-computes a threshold value $T$
as a function of the distributions $\F_i$, and
then hires the first candidate $i$ whose value passes the
threshold, i.e.~$X_i \ge T$.  \cite{SamuelCahn} was the
first to describe such a stopping rule which uses as a
threshold the median of the distribution of $\max_{i = 1}^n
X_i$, in other words, they define $T_{\mathrm{SC}}$ such
that $\Pr[\max_{i = 1}^n X_i \ge T_{\mathrm{SC}}] =
1/2$, and \cite{KleinbergWeinberg}
use $T_{\mathrm{KW}} = \frac12 \expect[\max_{i =
1}^n X_i]$. Either threshold yields a stopping rule
satisfying a prophet inequality for the optimal constant
$c = 1/2$.

The following notation is going to be useful in the sections
to follow: for each candidate $i$, let $\supp(\F_i) = \{ x >
0 \mid \Pr[X_i = x] \ne 0 \}$ denote the support of
distribution $\F_i$ and let $\Support = \bigcup_{i = 1}^n
\supp(\F_i)$. For the remainder of the paper we assume, with
possible exception some expository examples, that every
$\supp(\F_i)$ is a \emph{discreet} set.
%and further, that every variable $X_i$ is
%supported on common set, i.e.~$\Support =
%\supp(\F_1) = \ldots = \supp(F_n)$.

\section{Fairness}
\label{sec:fairness}

In the introduction we studied examples of different kinds
of unfairness that manifest in the solutions of optimal
stopping problems and pinned down two particular kinds of
fairness criteria that were violated in each case. Here we
give formal definitions for those fairness properties in the
framework of the hiring problem we defined in the previous
section.

We begin by considering the situation in
\Cref{ex:rideshare}. The observation there was that
conditional on candidate $i$ having value $x$, the
probability of them getting hired should only depend on the
value $x$ and not on the identity $i$ of the particular
candidate. This leads naturally to the following fairness
definition adapted to our framework.

\begin{definition}[Identity-Independent Fairness (IIF)]
\label{defn:iif}
A algorithm $\ALG$ for a given instance $\I = ((\F_i)_{i =
1}^n, \pi)$ supported on $\Support$ is said to satisfy
\emph{Identity-Independent Fairness} if there exist a
function $p : \Support \to [0, 1]$ such that:

\[ \Pr[ \ALG \text{ hires } i \mid X_i = x ] =
p(x),\quad \forall i \in [n], x \in \Support \]

In other words, the probability of hiring candidate $i$
conditional on their value being $x$ is independent of their
\emph{identity} $i$ (but potentially dependent on their
value $x$).

When an algorithm/stopping rule satisfies this property, we
say it is ``IIF with probability function $p(x)$''.
\end{definition}

The above definition is given for an instance with a
particular arrival ordering $\pi$ but one can extend the
definition to say that a \emph{family} of algorithms
$\{\ALG^{\pi}\}_{\pi \in S_n}$ is IIF if each algorithm in
the family is IIF with some probability function
$p_{\pi}(x)$ (potentially different for each permutation).

Notice that \Cref{defn:iif} applies even to offline
algorithms --- ones that, just like the prophet, are not
constrained to make a decision before seeing the entirety of
the input.

Similarly, we can formalize the fairness notion demonstrated
in \Cref{ex:firm-hiring} as follows.

\begin{definition}[Time-Independent Fairness (TIF)]
\label{defn:tif}
A family of algorithms $\{\ALG^{\pi}\}_{\pi \in S_n}$, one for each
arrival ordering $\pi$, for an instance $\I = (\F_i)_{i =
1}^n$ supported on $\Support$ is said to satisfy
\emph{Time-Independent Fairness} if there exists a function
$p : [n] \times \Support \to [0, 1]$ such that:

\[ \Pr[\ALG^{\pi} \text{ hires } i \mid X_i = x] = p(i,
x),\quad \forall i \in [n], x \in \Support, \pi \in S_n \]

In other words, the probability of hiring candidate $i$
conditional on their value being $x$ is independent of their
arrival time across different arrival orderings (but is
allowed to depend on their identity $i$ and the value $x$).
\end{definition}

Unlike \Cref{defn:iif}, this definition of fairness is
meaningful only in the online setting. When considering the
behavior of an offline algorithm with access to all r.v.s at
once, the definition is trivially satisfied.

We claim that the two definitions capture different notions
of fairness and indeed they are independent of one another.
The following examples demonstrate that fact by giving
instances and families of stopping rules which satisfy
either \Cref{defn:iif} or \Cref{defn:tif} but not both.

\begin{itemize}
\item Consider an instance with two r.v.s $X_1, X_2$ of
arbitrary but independent distributions. For arrival
permutation $\pi_1 = (1, 2)$, define $\ALG^{\pi_1}$ as the
stopping rule which hires one of $X_1, X_2$ uniformly at
random. On the other hand, for $\pi_2 = (2, 1)$, define
$\ALG^{\pi_2}$ as the stopping rule that terminates without
hiring.

Each algorithm in the family $\{\ALG^{\pi_i}\}_{i = 1}^2$
satisfies the IIF definition with $p_{\pi_1}(x) = 1/2,
p_{\pi_2}(x) = 0$ respectively but it does {\bf not} satisfy
the TIF definition because of the dependence on $\pi$.

\item Consider again an instance $X_1, X_2$ of two
independent and arbitrarily distributed r.v.s. For either
arrival permutation, define a stopping rule $\ALG^{\pi_i}$
which always hires candidate $i = 1$ regardless of the
realization of their value and regardless of their arrival
time.

Then, the family $\{\ALG^{\pi_i}\}_{i = 1}^n$ satisfies the
TIF property with $p(i, x) = p(i) = \ind[i = 1]$. On the
other hand, none of the algorithms in the family satisfy the
IIF property because of the dependence on $i$.
\end{itemize}

More importantly, the definitions are {\bf not} mutually exclusive
and there exist families of algorithms that satisfy both.
A trivial example is an algorithm that never hires anyone.
This is obviously IIF with $p(x) = 0$ for all $x$ as well as
TIF with $p(i, x) = 0$ and indeed can be described as fair
even if it performs poorly.
Another example is an algorithm which hires a candidate $i$
uniformly at random among all the candidates regardless of
the permutation $\pi$ and regardless of the realization of
the value $X_j$ of any candidate $j$. This algorithm
is both IIF and TIF but again has poor performance tending
to zero as $n \to \infty$.

\maedit{
Before moving on, let's also verify what was implied earlier
in the introduction, that common stopping rules for the
hiring problem indeed fail to satisfy both of our fairness
definitions.
Consider a simple instance with two 0-1 random variables
distributed as follows:
}

\maedit{
\[
X_1 =
\begin{cases}
1, & \wp 1/2\\
0, & \wp 1/2
\end{cases}
,\quad
X_2 = 
\begin{cases}
1, & \wp 2/3\\
0, & \wp 1/3
\end{cases}
\]
}

\maedit{
Let $\ALG^{\rightarrow}$ be any single-threshold stopping
rule which uses a threshold $T \in (0, 1)$ (e.g.~the
\cite{KleinbergWeinberg} or the \cite{SamuelCahn} stopping
rule) on the forward arrival
ordering, i.e.~$\pi = (1, 2)$.
%The common threshold used to make
%hiring decisions is $T = \frac12 \cdot \expect[ \max(X_1,
%X_2) ] = \frac{5}{12}$
In this instance, such a rule would
hire a candidate under consideration if and only if their
value is equal to $1$. Similarly, $\ALG^{\leftarrow}$ is
the same stopping rule but applied on the reverse arrival
ordering, i.e.~$\pi = (2, 1)$. Consider now the hiring
probabilities of each r.v.~and on each arrival ordering
conditional on sampling the higher support point:
}

\maedit{
\begin{align*}
\Pr[ \ALG^{\rightarrow} \text{ hires } 1 \mid X_1 = 1] =
1,\quad&
\Pr[ \ALG^{\rightarrow} \text{ hires } 2 \mid X_2 = 1] =
\Pr[X_1 = 0] = \frac12\\
\Pr[ \ALG^{\leftarrow} \text{ hires } 1 \mid X_1 = 1] =
\Pr[X_2 = 0] = \frac13,\quad&
\Pr[ \ALG^{\leftarrow} \text{ hires } 2 \mid X_2 = 1] = 1
\end{align*}
}

\maedit{
Fixing an arrival ordering, we see that neither stopping
rule is IIF because the conditional probabilities depend on
the identity $i$ of the variable considered. Further, the
family $\{\ALG^{\rightarrow}, \ALG^{\leftarrow}\}$ fails to
satisfy the TIF property since there is a dependence on the
arrival ordering.
}

%\macomment{Highlight somehow this paragraph. It is
%important because it highlights our plan moving forward.}
As we observed earlier, the IIF property can be desirable even
in an offline version of the hiring problem. It is therefore
possible to think of fairness as an extra dimension of
complexity on top of the online dimension of the problem.
Each dimension (offline vs.~online and unfair vs.~fair)
constrains the set of algorithms that can be used in
different ways and degrade performance as measured by the
expected value of the person hired. Akin to a prophet
inequality, one can ask what is the best competitive ratio
achievable when comparing algorithms that have to
satisfy different kinds of constraints. For example, how
much more powerful are offline IIF algorithms (IIF prophets)
compared to IIF stopping rules? Or, how much more powerful
are general stopping rules without fairness considerations
compared to IIF stopping rules? What about compared to TIF
stopping rules?

\macomment{Maybe give the extended definition of competitive
ratio here.}

In the sections that follow, we address those
questions for interesting combinations of settings in order
to understand the trade-offs in performance that arise
when different kinds of constraints are imposed on the
decision maker.

\section{A Characterization of Fair Stopping Rules}
\label{sec:characterization}

We now turn our attention to designing optimal, fair
(IIF/TIF) stopping rules, which will eventually lead us to a
unified approach to designing algorithms with good
competitive ratios in multiple settings. The reason this is
a good starting point is that, as we shall see shortly, the
combined requirement for fairness along with the online
nature of the setting induces a good structure to the
set of available algorithms allowing us to compute the
optimal by solving a simple polynomial-size Linear Program
(LP).

\subsection{IIF Constraints}

The following lemmas formalize the observation that the probability
functions in the definition of the IIF property characterize
both the performance as well as the structure of the
stopping rule satisfying that property.

\begin{lemma}
\label{lem:iif-obj}
Consider an instance $\I = ((\F_i)_{i = 1}^n, \pi)$ with
an attached arrival ordering.
All stopping rules $\ALG$ for this instance
satisfying the IIF property with probability function
$p(x)$ have the \emph{same} expected performance which is
given by the expression:

\begin{equation}
\label{eq:iif-obj}
\expect[ \ALG ] = \sum_{i = 1}^n \sum_{x \in \Support}
x \cdot f_i(x) \cdot p(x)
\end{equation}
\end{lemma}

\begin{proof}
\begin{align*}
\expect[\ALG]
&= \expect\left[ \sum_{i = 1}^n X_i \cdot \ind[\ALG \text{ hires
} i]\right]\\
&= \sum_{i = 1}^n \sum_{x \in \Support} \expect[X_i \cdot
\ind[\ALG \text{ hires } i] \mid X_i = x] \cdot \Pr[X_i =
x]\\
&= \sum_{i = 1}^n \sum_{x \in \Support} x \cdot \Pr[\ALG
\text{ hires } i \mid X_i = x] \cdot \Pr[X_i = x]\\
&= \sum_{i = 1}^n \sum_{x \in \Support} x \cdot p(x)
\cdot f_i(x)
\end{align*}
\end{proof}

\begin{algorithm}
\caption{IIF Stopping Rule}
\label{alg:iif-lp}
\SetKwInOut{Input}{Parameters}
\Input{$\I = ( (\F_i)_{i = 1}^n, \pi ), p(\cdot)$}
\For{$t = 1, \ldots, n$}{
$i \gets \pi(t)$.\\
{\bf Inspect} $X_i$ and let $x \gets X_i$.\\
$Q_t \gets 1 - \displaystyle \sum_{k = 1}^{t-1}
\sum_{y \in \Support} f_{\pi(k)}(y) \cdot p(y)$\\
Flip a coin with Heads probability equal to 
$\displaystyle q_t(x) = \frac{p(x)}{Q_t}$.\\
\uIf{coin comes up \emph{Heads}}{
{\bf Hire} $i$ and \emph{halt}.
}\uElse{
{\bf Reject} and proceed.
}
}
\end{algorithm}

\begin{lemma}[IIF Structural Lemma]
\label{lem:iif-structural}
A stopping rule $\ALG$ for an instance $\I =
((\F_i)_{i = 1}^n, \pi)$ is IIF with probability function
$p(x)$ if and only if $p$ satisfies:

\begin{equation}
\label{eq:iif-ineq}
p(x) + \sum_{k = 1}^{n-1} \sum_{y \in \Support}
p(y) \cdot f_{\pi(k)}(y) \le 1,\quad \forall x \in
\Support
\end{equation}

Moreover, if $p$ is a function satisfying the above
inequality, then \Cref{alg:iif-lp} (parameterized by
$p$) satisfies the IIF property with the same
probability function $p$.
\end{lemma}

\begin{proof}
For the ``only if'' direction, the key observation is that
for the probabilities $p(x)$ to be valid, they have to
be such that the probability of hiring $X_{\pi(t)}$
conditional on sampling any value $x \in \Support$ is
bounded above by the probability of reaching time step
$t$, because you cannot hire someone with probability
greater than that of reaching them in the first place.

More formally, let $Q_t$ be the probability of
reaching time step $t$ computed as follows:

\begin{align*}
Q_t &= \Pr[\ALG \text{ rejects } \pi(1),
\ldots, \pi(t-1)]
= 1 - \Pr\left[ \bigcup_{k = 1}^{t-1} \ALG \text{ accepts
} \pi(k) \right]\\
&= 1 - \sum_{k = 1}^{t-1} \sum_{y \in \Support}
\Pr[\ALG \text{ accepts } \pi(k) \mid X_{\pi(k)} =
y ] \cdot \Pr[ X_{\pi(k)} = y ]\\
&= 1 - \sum_{k = 1}^{t-1} \sum_{y \in \Support} p(y)
\cdot f_{\pi(k)}(y)
\end{align*}

\noindent
where we used the fact that a stopping rule hires at most
one candidate to assert that the events in the union are
disjoint.

We thus need to require $Q_t \ge p(x)$ for all $t \in
[n], x \in \Support$. Notice however that $Q_t$ is
non-negative, decreasing in $t \in [n]$ and independent of
$x$ therefore it suffices to require the inequality for $t =
n$, which gives the required condition in the statement of
the Lemma:

\[ p(x) \le 1 - \sum_{k = 1}^{n-1} \sum_{y \in
\Support} p(y) \cdot f_{\pi(k)}(y),\quad \forall x \in
\Support \]

For the ``if'' direction, it suffices to prove that
\Cref{alg:iif-lp} satisfies the IIF property with
probability function $p(x)$ if we are given that $p$
satisfies inequalities (\ref{eq:iif-ineq}).
We do this by induction on the time step $t \in [n]$.
For the base case, $Q_1$ is set to $1$ in the
algorithm, therefore $q_1(x) =
p(x)$ and we're done.
For the inductive step, consider
any time step $t > 1$.
The conditional probability of hiring
candidate $\pi(t)$ is as follows:

\begin{align*}
\Pr[\text{\Cref{alg:iif-lp} hires } \pi(t) \mid
X_{\pi(t)} = x ]
=& \Pr[\text{\Cref{alg:iif-lp} hires } \pi(t) \mid
X_{\pi(t)} = x, \text{ have reached time step } t]\\
&\cdot \Pr[\text{\Cref{alg:iif-lp} rejects } \pi(1),
\ldots, \pi(t-1)]
\end{align*}

The first factor in the product is exactly the probability
$q_t(x)$ of the coin used by \Cref{alg:iif-lp}.
Using the inductive hypothesis, and a similar computation to
the ``only if'' direction, we can express the probability of
\Cref{alg:iif-lp} rejecting the first $t-1$ candidates in
terms of the p.d.f.s of each $X_i$ and the probabilities
$p(x)$, getting $\Pr[\text{\Cref{alg:iif-lp} rejects }
\pi(1), \ldots, \pi(t-1)] = Q_t$, where $Q_t$ is
exactly the quantity computed in the for-loop.
Hence,
$ \Pr[\text{\Cref{alg:iif-lp} hires } X_{\pi(t)} \mid
X_{\pi(t)} = x] = q_t(x) \cdot Q_t = \frac{p(x)}{Q_t}
\cdot Q_t = p(x)$
\end{proof}

A consequence of the structural lemma is that it allows us
to reduce the problem of designing a fair stopping rule into
a simple LP.

Indeed, for IIF stopping rules, \Cref{lem:iif-structural}
states that the set of all functions $p(x)$ for which there
exist an IIF stopping rule with that probability function
matches exactly the set of all functions $p(x)$ that satisfy
inequalities (\ref{eq:iif-ineq}). Further,
\Cref{lem:iif-obj} states that the optimal stopping rule
would be the one maximizing a linear function of the
probabilities $p(x)$. Therefore, the solution to the
following Linear Program gives the conditional hiring
probabilities of the \emph{optimal, IIF} stopping rule for a
given instance $\I = ((\F_i)_{i = 1}^n, \pi)$:

\begin{equation}
\tag{OPT Online IIF}
\label{lp:online-iif}
\left\{
\begin{array}{lcl}
\max & \displaystyle \sum_{i = 1}^n \sum_{x \in S} x \cdot
 f_i(x)  \cdot p(x) &\\
\text{s.t.} & \displaystyle p(x) + \sum_{k = 1}^{n-1}
\sum_{y \in S} f_{\pi(k)}(y) \cdot p(y)\le 1, & \forall x \in S.\\
& p(x) \in [0, 1] & \forall x \in S.
\end{array}
\right.
\end{equation}

After solving (\ref{lp:online-iif}), one only has to
use the solution to parameterize \Cref{alg:iif-lp} which
becomes an optimal, fair stopping rule. We have thus proved
the following theorem.

\begin{theorem}
The optimal IIF stopping rule for a given instance $\I =
((\F_i)_{i = 1}^n, \pi)$ can be computed in time polynomial
in $n, |\Support|$ by solving the (\ref{lp:online-iif}).
\end{theorem}

\subsection{TIF Constraints}

Similarly to IIF fairness, we proceed to prove similar
structural lemmas for TIF stopping rules.

\begin{lemma}
Consider an instance $\I = (\F_i)_{i = 1}^n$.  Any
stopping rule $\ALG^{\pi}$ for $\I$ satisfying the TIF
property with probability function $p(i, x)$ has the
expected performance given by the expression:

\[ \expect[ \ALG^{\pi} ] = \sum_{i = 1}^n \sum_{x \in \Support}
x \cdot f_i(x) \cdot p(i, x) \]

which is independent of the specific permutation $\pi$ and
the specifics of the stopping rule.
\end{lemma}

The proof of the lemma is omitted as it uses the same
probabilistic calculations as \Cref{lem:iif-obj}, the only
difference being that $p(i, x)$ is now in place of
$p(x)$.

\begin{algorithm}
\caption{TIF Stopping Rule}
\label{alg:tif-lp}
\KwData{$\I = ((\F_i)_{i=1}^n, \pi), p(\cdot, \cdot)$}
\For{$t = 1, \ldots, n$}{
$i \gets \pi(t)$.\\
{\bf Inspect} $X_i$ and let $x \gets X_i$.\\
$Q^{\pi}_t \gets 1 - \displaystyle \sum_{k = 1}^{t-1}
\sum_{y \in \Support} f_{\pi(k)}(y) \cdot p(\pi(k), y)$\\
Flip a coin with Heads probability equal to 
$\displaystyle q^{\pi}_t(x) = \frac{p(i, x)}{Q^{\pi}_t}$.\\
\uIf{coin comes up \emph{Heads}}{
{\bf Hire} $i$ and halt.
}\uElse{
{\bf Reject} $i$ and proceed.
}
}
\end{algorithm}

\begin{lemma}[TIF Structural Lemma]
\label{lem:tif-structural}
A family of stopping rules $\{\ALG^{\pi}\}_{\pi \in
S_n}$ for an instance $\I = (\F_i)_{i = 1}^n$ satisfies
the TIF property with probability function $p(i, x)$ if and
only if $p$ satisfies:

\begin{equation}
\label{eq:tif-no-perm}
p(i, x) + \sum_{k \neq i} \sum_{y \in \Support} p(k, y)
\cdot f_{k}(y) \le 1,\quad \forall i \in [n], x \in \Support
\end{equation}

Moreover, if $p$ is a function satisfying the above
inequality, then the family defined by \Cref{alg:tif-lp} for
all permutations $\pi \in S_n$ satisfies the TIF property
with the same probability function $p$.
\end{lemma}

\begin{proof}
For the ``only if'' direction we proceed as in the IIF case
by first expressing the probability of rejecting the first
$t-1$ candidates under some arrival ordering $\pi$.

\begin{align*}
Q^{\pi}_t &= \Pr[\ALG^{\pi} \text{ rejects } \pi(1),
\ldots, \pi(t-1)]
= 1 - \Pr\left[ \bigcup_{k = 1}^{t-1} \ALG^{\pi} \text{
  accepts } \pi(k) \right]\\
&= 1 - \sum_{k = 1}^{t-1} \sum_{y \in \Support}
\Pr[\ALG^{\pi} \text{ accepts } \pi(k) \mid X_{\pi(k)} =
y ] \cdot \Pr[ X_{\pi(k)} = y ]\\
&= 1 - \sum_{k = 1}^{t-1} \sum_{y \in \Support} p(\pi(k), y)
\cdot f_{\pi(k)}(y)
\end{align*}

\noindent
where we used the fact that a stopping rule hires at most
one candidate to assert that the events in the union are
disjoint.

We then require that the conditional hiring probability is
bounded by the arrival probability at time step $t$:

\begin{equation}
\label{eq:tif-perm}
p(\pi(t), x) \le Q^{\pi}_t,\quad \forall t \in
[n], \forall x \in \Support \text{ and } \forall \pi \in
S_n
\end{equation}

Unlike the IIF case which considered a single arrival
permutation, the TIF property applies to the whole family of
stopping rules and the constraint we just described has to
apply for every such permutation. However, we can simplify
it and prove that the set of inequalities in
(\ref{eq:tif-no-perm}) is satisfied if and only if the set of
inequalities (\ref{eq:tif-perm}) is satisfied.
The direction ``$(\ref{eq:tif-perm}) \Rightarrow
(\ref{eq:tif-no-perm})$'' follows easily by taking $\pi$ to be
any permutation such that $\pi(n) = i$. For the other
direction, assume inequalities (\ref{eq:tif-no-perm}) are all
satisfied, let $\pi \in S_n$ be arbitrary permutation, 
fix arbitrary $t \in [n], x \in \Support$ and denote $i
= \pi(t)$. Define permutation $\pi'$ which is derived from
$\pi$ by swapping the $t$-th and $n$-th elements. More
precisely,

\[ \pi'(j) = \left\{
  \begin{array}{ll}
    \pi(n), & j = t\\
    i = \pi(t), & j = n\\
    \pi(j),& \text{o.w.}
  \end{array}
  \right.
\]

Inequality (\ref{eq:tif-no-perm}) for $i, x$ states that:
$ p(i, x) \le 1 - \sum_{j \neq i} \sum_{y \in \Support}
p(j, y) \cdot f_j(y)$
or, equivalently expressed using permutation $\pi'$,
$ p(\pi'(n), x) \le 1 - \sum_{k = 1}^{n-1} \sum_{y \in \Support}
p(\pi'(k), y) \cdot f_{\pi'(k)}(y) $

The right hand side of the above inequality cannot
decrease if we make the outer summation range from $1$ to
$t-1$ instead so,
$ p(\pi'(n), x) \le 1 - \sum_{k = 1}^{t-1} \sum_{y \in \Support}
p(\pi'(k), y) \cdot f_{\pi'(k)}(y) $

Since $\pi'$ agrees with $\pi$ for all $j < t$, the above is
equivalent to:
$ p(\pi(t), x) \le 1 - \sum_{k = 1}^{t-1} \sum_{y \in \Support}
p(\pi(k), y) \cdot f_{\pi(k)}(y) $
which is exactly inequality (\ref{eq:tif-perm}).

Now for the ``if'' direction of the lemma, we prove that
the family defined by \Cref{alg:tif-lp} for all $\pi \in
S_n$ satisfies the TIF property with
probability function $p(i, x)$ again using induction
on the time step $t \in [n]$.

The base case is again trivial with $Q^{\pi}_1 = 1$
therefore $q^{\pi}_1(x) = p(i, x)$.

For the inductive step, consider
time step $t > 1$.
The conditional probability of $\ALG^{\pi}$ hiring
candidate $\pi(t)$ is:

\begin{align*}
\Pr[\text{\Cref{alg:tif-lp} hires } \pi(t) \mid
X_{\pi(t)} = x ]
=& \Pr[\text{\Cref{alg:tif-lp} hires } \pi(t) \mid
X_{\pi(t)} = x, \text{ have reached time step } t]\\
&\cdot \Pr[\text{\Cref{alg:tif-lp} rejects } \pi(1),
\ldots, \pi(t-1)]
\end{align*}
The first factor in the product is exactly the probability
$q^{\pi}_t(x)$ of the coin used by \Cref{alg:tif-lp}.
Using the inductive hypothesis, we can express the probability of
\Cref{alg:tif-lp} rejecting the first $t-1$ candidates in
terms of the p.d.f.s of each $X_i$ and the probabilities
$p(i, x)$, getting $\Pr[\text{\Cref{alg:tif-lp} rejects }
X_{\pi(1)}, \ldots, X_{\pi(t-1)}] = Q^{\pi}_t$, where
$Q^{\pi}_t$ is exactly the quantity computed in the
for-loop.  Hence,
\[ \Pr[\text{\Cref{alg:tif-lp} hires }
\pi(t) \mid X_{\pi(t)} = x] = q^{\pi}_t(x) \cdot
Q^{\pi}_t = \frac{p(i, x)}{Q^{\pi}_t} \cdot Q^{\pi}_t = p(i,
x) \]
\end{proof}

Similarly to the IIF case, we can express the probability
function $p(i, x)$ of the optimal family of TIF stopping
rules as the solution to the following LP:

\begin{equation}
\tag{OPT Online TIF}
\label{lp:online-tif}
\left\{
\begin{array}{lcl}
\max & \displaystyle \sum_{i = 1}^n \sum_{x \in \Support} x
\cdot p(i, x) \cdot f_i(x) &\\
\text{s.t.} & \displaystyle p(i, x) + \sum_{k \ne i} \sum_{y
\in \Support} f_k(y) \cdot p(k, y) \le 1, & \forall i \in
[n], \forall x \in \Support.\\
& p(i, x) \in [0, 1] & \forall i \in [n], \forall x \in
\Support.
\end{array}
\right.
\end{equation}

\begin{theorem}
The optimal TIF family of stopping rules for a given instance
$\I = (\F_i)_{i = 1}^n$ can be
computed in time polynomial in $n, |\Support|$ by solving
(\ref{lp:online-tif}).
\end{theorem}

\subsection{Offline Relaxation}
\label{sec:offline}

The following LP relaxation of the offline problem
(i.e.~prophet) is going to be useful in later sections.

\begin{lemma}
Let $\I = (\F_i)_{i = 1}^n$ be an instance of the hiring
problem. Let $\OptOff$ be the value of the candidate hired
by the optimal offline algorithm for instance $\I$ and let
$C^*$ be the optimal solution to the following LP:

\begin{equation}
\tag{Offline Relaxation}
\label{lp:offline}
\left\{
\begin{array}{lcl}
\max & \displaystyle \sum_{i = 1}^n \sum_{x \in S} x \cdot
p_{ix} \cdot f_i(x) &\\
\text{s.t.} & \displaystyle \sum_{i = 1}^n \sum_{x \in S}
p_{ix} \cdot f_i(x) \le 1\\
& p_{ix} \in [0, 1] & \forall i \in [n], \forall x \in S.
\end{array}
\right.
\end{equation}

Then, \[ \expect[\OptOff] \le C^* \]
\end{lemma}
\begin{proof}
Let $p_{ix} = \Pr[\OptOff \text{ hires } i \mid X_i = x]$.
Since any algorithm hires at most one candidate, it follows
that $\expect[\text{\# of candidates hired by } \OptOff] \le
1$. The expectation can be expressed in terms of $p_{ix}$ as
follows:
\begin{align*}
\expect[\text{\# of candidates hired by } \OptOff]
&= \expect\left[\sum_{i = 1}^n \ind[\OptOff \text{ hires }
i]\right]\\
&= \sum_{i = 1}^n \sum_{x \in \Support}
\Pr[\OptOff \text{ hires } i \mid X_i = x]
\cdot \Pr[X_i = x]\\
&= \sum_{i = 1}^n \sum_{x \in \Support} p_{ix} \cdot f_i(x)
\end{align*}

Therefore, the probabilities $p_{ix}$ of $\OptOff$
constitute a feasible solution to the LP and further the
objective value $C$ of the LP for $p_{ix}$ is equal to the
expected performance of $\OptOff$.  Hence, the optimal
solution $p_{ix}^*$ with objective value $C^*$ cannot be any
worse: $C^* \ge C = \expect[\OptOff]$. 
\end{proof}

The special form of this LP allows to prove that there are
optimal solutions with special structure that will allow us
later to transform them into solutions to Online IIF/TIF LPs
presented previously.

\begin{lemma}
\label{lem:offline-structural}
Among the set of optimal solutions to
(\ref{lp:offline}), there is one that
satisfies $p_{ix} = p_{jx}$ for all
$i,j \in [n]$ and all $x \in \Support$.
\end{lemma}
\begin{proof}
    Let $p = (p_{ix})$ be any optimal solution
    to (\ref{lp:offline}). We will explain 
    how to construct another solution 
    $p' = (p'_{ix})$ that obeys the
    LP constraints and has the same 
    objective value. For any $x \in S$,
    such that $\sum_{i=1}^n p_{ix} f_{ix} = 0$
    % (in other words, $p_{ix}=0$ for all $i$)
    we set $p'_{ix}=0$. For any
    $x \in S$ such that 
    $\sum_{i=1}^n p_{ix} f_i(x) = w_x > 0$,
    we let $z_x = \sum_{i=1}^n f_i(x)$ and set
    $p'_{ix} = w_x/z_x$ for all $i$.

    Note that
    $0 \le p'_{ix} \le 1$ because 
    $w_x/z_x$ is the weighted average
    of the numbers $p_{ix}$, weighted by
    $f_i(x)/z_x$, and each $p_{ix}$ belongs
    to $[0,1]$.

    The one constraint
    of (\ref{lp:offline}) %(the sum constraint)
    is satisfied by $p'$ because 
    \begin{equation}
    \label{eq:structural-1}
    \sum_{i=1}^n p'_{ix} f_i(x) = 
    \frac{w_x}{z_x} \sum_{i=1}^n f_i(x) =
    \frac{w_x}{z_x} \cdot z_x = w_x =
    \sum_{i=1}^n p_{ix} f_i(x)
    \end{equation}
    \noindent
    and summing the above equations over $x$,
    we get that the LHS of the LP constraint
    is the same for $p'$ as for $p$.

    %The objective function value is the
    %same for both solutions, for the 
    %same reason.

    Finally, to verify the 
    equality of the objective function
    values, scale \Cref{eq:structural-1}
    %equation
    %$\sum_{i=1}^n p'_{ix} f_i(x) = 
    %\sum_{i=1}^n p_{ix} f_i(x)$
    by $x$, then sum over all $x \in S$.
\end{proof}

\section{Competitive Fair Algorithms}
\label{sec:competitive}

In the previous sections we described how to reduce the
problem of computing the optimal IIF/TIF stopping rule for a
given instance to solving an LP.  Here, we use this
reduction to show tight competitive ratios between different
settings (online vs.~offline, unfair vs.~fair etc.). In what
follows, we present the results separately for each kind of
fairness. \maedit{The results are summarized in
\Cref{fig:squares} as competitive ratios between different
settings depicted on top of arcs connecting those settings
and our theorems indicate which arc they correspond to.}

\subsection{Competitive IIF}

\begin{figure}
\centering
\begin{subfigure}{0.5\textwidth}
\begin{tikzpicture}
    \node[draw=none] (A) at (1.5,1.5) {On};
    \node[draw=none] (B) at (1.5, -1.5) {On,IIF};
    \node[draw=none] (C) at (-1.5, -1.5) {Off,IIF};
    \node[draw=none] (D) at (-1.5, 1.5) {Off};

    % On -- On,IIF : Right
    \path [->] (A) edge node[right] {$=\frac12$} (B);
    % Off,IIF -- On, IIF : Bottom
    \path [->] (C) edge node[below] {$=\frac12$} (B);
    % Off -- Off,IIF : Left
    \path [->] (D) edge node[left] {$\ge \frac12$} (C);
    % Off -- On : Top
    \path [->] (D) edge node[above] {$=\frac12$} (A);

    % Off -- On,IIF : Diagonal
    \path [->] (D) edge node[above] {$=\frac12$} (B);
\end{tikzpicture}
\end{subfigure}
\hfill
\begin{subfigure}{0.4\textwidth}
\begin{tikzpicture}
    \node[draw=none] (A) at (1.5,1.5) {On};
    \node[draw=none] (B) at (1.5, -1.5) {On,TIF};
    \node[draw=none] (D) at (-1.5, 1.5) {Off};

    % On -- On,TIF : Right
    \path [->] (A) edge node[right] {$=\frac12$} (B);
    % Off -- On : Top
    \path [->] (D) edge node[above] {$=\frac12$} (A);

    % Off -- On,TIF : Diagonal
    \path [->] (D) edge node[above] {$=\frac12$} (B);
\end{tikzpicture}
\end{subfigure}
\caption{Summary of different \emph{settings} we consider.
Labels on arcs represent competitive ratios: an
arc from setting $A$ to setting $B$ labeled ``$\ge r$''
(resp.~``$= r$'') means a competitive ratio bound of the form:
for any instance $(\F_i)_{i = 1}^n$, let $\OPT^A, \OPT^B$
be the optimal algorithms on this instance for the
respective settings, then
$\expect\left[\OPT^B\right] / \expect\left[\OPT^A\right] \ge
r$ (resp.~the bound is tight). The settings are named in the
form $X[,Y]$, where $X \in \{\text{Off}, \text{On}\}$
denotes whether it is an offline or online setting and $Y
\in \{\text{IIF}, \text{TIF}\}$ (if present)
represents the kind of fairness property required (if any).}
\label{fig:squares}
\end{figure}
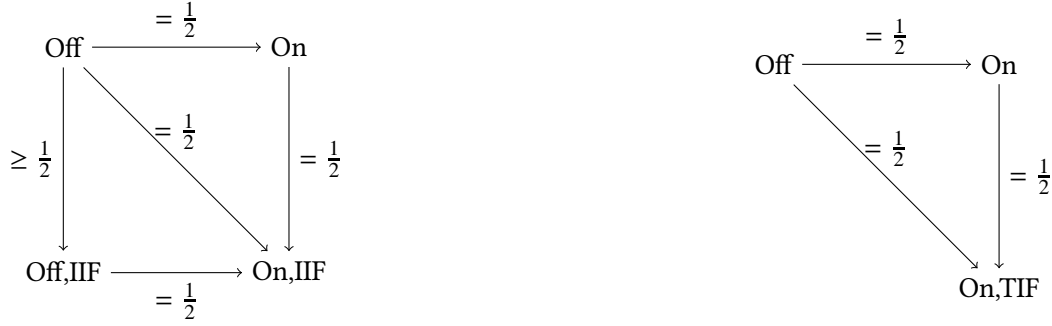

For the IIF property, there are four settings of interest
that correspond to all combinations of offline/online and
non-fair/fair settings according to the IIF property. We
pictorially represent each of those settings as vertices on
a square in the left side of \Cref{fig:squares}. Edges on
this diagram represent competitive ratios between the
corresponding settings (either already known or proven in
this paper).

The top edge of the IIF diagram is the well-known prophet
inequality we described in \Cref{sec:prelim} \maedit{which,
along with the tight example, gives a competitive ratio of
$1/2$ between the unconstrained (in terms of fairness)
offline setting and an unconstrained online setting}.

Next, we focus of the diagonal edge comparing the offline
non-fair setting (i.e.~a prophet) with an online and IIF
setting for which we manage to recover a $1/2$-competitive
ratio that is tight. This means that despite the fact that a
decision maker in the Online IIF setting is more restrained
as to what decisions they can make in order to remain
fair, they nevertheless can perform at least half as good as
an omniscient prophet who makes no effort in maintaining
fairness!

\begin{theorem}[{\maedit{IIF Diagonal Arc}}]
\label{thm:main1}
Let $\I = ((F_i)_{i = 1}^n, \pi)$ be an instance of the
hiring problem. Let $\OptOff$ be the optimal, offline
algorithm on $\I$ (i.e.~the prophet choosing $i^* \in
\argmax_{i} X_i$), and $\OptOnIIF$ be the optimal IIF
stopping rule on the same instance. Then,

\[ \expect[ \OptOnIIF ] \ge \frac12 \cdot \expect[ \OptOff
] \]

Moreover, the inequality is tight.
\end{theorem}

\begin{proof}
We start by considering an optimal solution $p_{ix}^*$ to
(\ref{lp:offline}) presented in \Cref{sec:offline} which
satisfies $p_{ix}^* = p_{jx}^*$ for all $i, j$ and $x$ and
denote that common value by $p_x^*$. The
existence of such a solution is guaranteed by
\Cref{lem:offline-structural}. Let
$C^*$ be the objective value of that solution. Based on
this, we define $p(x) = p_{x}^*/2$ for all $x \in
\Support$.

We claim that $p(x)$ is a feasible solution to
(\ref{lp:online-iif}) with objective value $C^*/2$.
Once we prove this, the theorem follows since $C^*$ is an
upper bound to $\expect[\OptOff]$.

Indeed, consider \maedit{the left-hand-side (LHS)
of a constraint of (\ref{lp:online-iif})
for some $x \in \Support$}:

\begin{align}
\label{eq:bound1}
p(x) + \sum_{k = 1}^{n-1} \sum_{y \in \Support}
f_{\pi(k)}(y) \cdot p(y)
&= \frac12 p_x^* + \frac12 \sum_{k = 1}^{n-1} \sum_{y \in
\Support} f_{\pi(k)}(y) \cdot p_y^*
\end{align}

Recall that $p_x^*$ as a solution to (\ref{lp:offline}) it
satisfies the feasibility constraint
$\sum_{i = 1}^n \sum_{x \in \Support} f_i(x) \cdot p_x^* \le
1$. Also, $p_x^* \in [0, 1]$, therefore, we can bound the
right-hand-side (RHS) of \Cref{eq:bound1} as follows:

\begin{align*}
\label{eq:bound1}
%p(x) + \sum_{k = 1}^{n-1} \sum_{y \in \Support}
%f_{\pi(k)}(y) \cdot p(y)
%&=
\frac12 p_x^* + \frac12 \sum_{k = 1}^{n-1} \sum_{y \in
\Support} f_{\pi(k)}(y) \cdot p_y^*
&\le \frac12 + \frac12 \sum_{k = 1}^n \sum_{x \in \Support}
f_{\pi(k)}(x) \cdot p_x^*\\
&= \frac12 + \frac12 \sum_{i = 1}^n \sum_{x \in \Support}
f_i(x) \cdot p_x^*
\le \frac12 + \frac12 = 1
\end{align*}

\maedit{As for the objective value of (\ref{lp:online-iif})
for $p(x)$}:

%\[ (\text{\ref{lp:offline}}) = \sum_{i = 1}^n \sum_{x \in
%\Support} x \cdot p_{x} \cdot f_i(x)
%= \frac12 \sum_{i = 1}^n \sum_{x \in \Support} x \cdot p(x)
%\cdot f_i(x) = \frac12 (\text{\ref{lp:online-iif}}) \]

\maedit{
\[
\sum_{i = 1}^n \sum_{x \in \Support} x \cdot f_i(x) \cdot
p(x)
= \frac12 \cdot \sum_{i = 1}^n \sum_{x \in \Support} x \cdot
f_i(x) \cdot p_x^*
= \frac{C^*}{2}
\]
}

The tightness of the bound follows from the tighness of the
standard prophet inequality with the tight instance we
presented in \Cref{sec:prelim}. To be precise though, we
need to adapt that example to our framework which requires
that all distributions share the same support. We therefore
consider the following $\delta$-perturbation of the standard
tight instance where we add probability mass $\delta <
\eps^2$ to the support points missing on the distribution of
each random variable:

\begin{equation}
\label{eq:ex}
X_1 = \left\{
  \begin{array}{ll}
  1 / \eps, & \wp \delta\\
  1,& \wp 1 - 2\delta\\
  0,& \wp \delta
  \end{array}
  \right.
  ,\quad
X_2 = \left\{
  \begin{array}{ll}
  1/\eps,& \wp \eps\\
  1,& \wp \delta\\
  0,& \wp 1 - \delta - \eps
  \end{array}
  \right.
\end{equation}

In this instance, a prophet gets $\expect[\OptOff] = 2 -
\bigO(\eps)$. On the other hand, any stopping rule (IIF or
not) on this instance with arrival ordering $\pi_1 = (1, 2)$
gets no more than $\expect[\text{Online}^{\pi_1}] = 1 +
\bigO(\eps)$.
\end{proof}

Regarding the remaining arcs on the IIF side of
\Cref{fig:squares}, notice that the $\ge 1/2$ part of each
bound is a consequence of \Cref{thm:main1}. This is because
the optimal algorithm for any setting in $\{\text{Off,IIF},
\text{On}\}$ performs no better than the optimal algorithm
in the Off setting (in expectation) and no worse than the
optimal in the On,IIF setting (in expectation). Having shown
that Off and On,IIF have a gap of $1/2$, it directly follows
that the competitive ratio between any two other setting is
at least $1/2$.

The formal results appear in the following theorems.

\begin{theorem}[{\maedit{IIF Bottom Arc}}]
\label{thm:iif-bottom-arc}
Let $\I = ((F_i)_{i = 1}^n, \pi)$ be an instance of the
hiring problem. Let $\OptOffIIF$ be the optimal, offline
\emph{and IIF} algorithm on $\I$, and $\OptOnIIF$ be the optimal,
IIF stopping rule on the same instance. Then,

\[ \expect[ \OptOnIIF ] \ge \frac12 \cdot \expect[
\OptOffIIF ] \]

Moreover, the inequality is tight.
\end{theorem}
\begin{proof}
We already argued how the inequality follows from
\Cref{thm:main1}.
To show tightness, consider again the instance in (\ref{eq:ex}).
We already argued that for $\pi_1 = (1, 2)$, no stopping rule
can do better than $1 + \bigO(\eps)$, therefore the same holds
for the any Online, IIF stopping rule. We now prove that
there is an Offline, IIF rule with expected performance
at least $2 - \bigO(\eps)$. To do this, it is easier to
design an online, IIF stopping rule for the reverse
permutation $\pi_2 = (2, 1)$, which implies the existence of
an Offline, IIF algorithm with the same expected
performance. Indeed let $\ALG^{\pi_2}$ work as follows:
Inspect $X_2$ and if $X_2 > 0$, accept it with probability
$1 - \eps$ (by tossing a biased coin). Never accept the
zero value for $X_2$. If $X_2$ was not hired, proceed to
inspect $X_1$. If $X_1 < 0$ again do not hire. Otherwise,
for any value $X_1 > 0$, accept it with probability $r = \frac{1
- \eps}{1 - q}$ where $q = \Pr[\ALG \text{ hires $X_2$ at
time step $t = 1$}] = (1 - \eps)\cdot(\eps + \delta)$. If
$\delta$ is sufficiently small, e.g.~$\delta < \eps^2$ then
$q \le \eps$ and so $r$ is a well-defined probability.

To verify the IIF property, consider the probability
$p(i, x) = \Pr[\ALG^{\pi_2} \text{ accepts } X_i \mid X_i =
x]$. If $x = 0$ then $p(i, x) = 0$ regardless of $i$.
For $x > 0$, it is $p(2, x) = 1 - \eps$ by definition and
$p(1, x) = \frac{1 - \eps}{1 - q} \cdot (1 - q) = 1 - \eps$.
Therefore $p(i, x) = p(x)$ confirming the definition of IIF.

Finally, the expected value of this IIF stopping rule applied on
$\pi_2$ can be computed as follows:

\begin{align*}
\sum_{i=1}^2 \sum_{x \in \{0, 1,1/\eps\}}
    x \cdot p(x) \cdot f_i(x)
&= 
(1-\eps)\cdot \expect[X_1] + (1-\eps)\cdot \expect[X_2]\\
&=
(1-\eps)c \cdot [1 + \delta(1/\eps - 2)]
+
(1-\eps)\cdot [1 + \delta]\\
&> 2 - 2\eps
\end{align*}

\end{proof}

\begin{theorem}[{\maedit{IIF Right Arc}}]
\label{thm:iif-right-arc}
Let $\I = ((F_i)_{i = 1}^n, \pi)$ be an instance of the
hiring problem. Let $\OptOn$ be the optimal stopping rule on
$\I$, and $\OptOnIIF$ be the optimal, IIF stopping rule on
same instance. Then,

\[ \expect[ \OptOnIIF ] \ge \frac12 \cdot \expect[
\OptOn
] \]

Moreover, the inequality is tight.
\end{theorem}
\begin{proof}
Again, we focus only on the upper bound and provide a tight
instance. Consider the following instance\footnote{For the
IIF property to be well-defined, technically we need a
common support across r.v.s. This can be amended by
considering a $\delta$-perturbation of the example just like
we did with previous tight examples.}:

\[
X_1 = 1,\quad X_2 =
\begin{cases}
  1, & \wp \eps\\
  0, & \wp 1 - \eps
\end{cases}
\]

with arrival ordering $\pi_1 = (1, 2)$. Clearly,
$\expect[\OptOn] = 1$ by always hiring the first candidate.
To compute the optimal, we solve (\ref{lp:online-iif}),
which in this case reduces to

\[
\begin{array}{ll}
\max & p_1 + \eps \cdot p_1\\
\st & p_0 + p_1 \le 1\\
    & p_1 + p_1 \le 1\\
    & p_0, p_1 \in [0, 1]
\end{array}
\]

%$\max (p_1 + \eps\cdot p_0)
%\text{ s.t. } \{ p_0 + p_1 \le 1, p_1 + p_1 \le 1, p_i \in [0,
%1] \}$

\noindent
whose optimal solution as $\eps \to 0$ approaches $1/2$,
meaning that $\expect[\OptOffIIF] \to 1/2$.
\end{proof}

Similarly, \Cref{thm:main1} implies a lower bound on the
competitive ratio for the arc on the left, comparing
fairness exclusively on an offline setting. For this, we do
not have a tight example and it is an interesting open
question whether this gap can be improved.

\begin{theorem}[{\maedit{IIF Left Arc}}]
\label{thm:iif-left-arc}
Let $\I = ((F_i)_{i = 1}^n, \pi)$ be an instance of the
hiring problem. Let $\OptOff$ be the optimal offline rule on
$\I$ (i.e.~a prophet), and $\OptOffIIF$ be the optimal, IIF
stopping rule on the same instance. Then,

\[ \expect[ \OptOffIIF ] \ge \frac12 \cdot \expect[
\OptOff ] \]
\end{theorem}

\subsection{Competitive TIF Algorithms}

We now move our attention to competitive ratios between
settings where the fairness criterion is the TIF property.
Here the picture is simpler since the TIF property only
makes sense in the online setting. The three interesting
settings we will be working with in this section are shown
on the right side of \Cref{fig:squares}. The top arc
again corresponds to the known, standard prophet inequality.
Here we prove the bounds for the other two arc using the
tools we developed in the previous sections.

We begin with the diagonal arc where again we manage to
prove a $1/2$-competitive ratio. In fact, the theorem that
follows is stronger as it implies the existence of a family
of algorithms that is at the same time TIF and each of the
algorithms individually is IIF achieving the desired ratio
of $1/2$. 

\begin{theorem}[{\maedit{TIF Diagonal Arc}}]
\label{thm:main2}
Let $\I = (F_i)_{i = 1}^n$ be an instance of the hiring
problem. Let $\OptOff$ be the optimal, offline algorithm
(i.e.~the prophet choosing $i^* \in \argmax_{i} X_i$).
There exists a family of TIF stopping rules
$\{\OnTIF^{\pi}\}_{\pi \in S_n}$ for instance $\I$ such
that,

\[ \expect[ \OnTIF^{\pi} ] \ge \frac12 \cdot \expect[ \OptOff
],\quad \forall \pi \in S_n \]

Moreover, each of $\OnTIF^{\pi}$ {\bf also satisfies} the
IIF property.
Finally, the inequality is tight.
\end{theorem}
\begin{proof}
Consider again an optimal solution $p_{ix}^*$ to
(\ref{lp:offline}) presented in \Cref{sec:offline}
\maedit{and denote the optimal objective value by $C^*$}.
Define $p(i, x) = p_{ix}^*/2$. Recall that $p_{ix}^*$ is
such that $\sum_{i = 1}^n \sum_{x \in \Support} p_{ix}^*
\cdot f_i(x) \le 1$.

To show that $p(i, x)$ as defined satisfies the constraints
of (\ref{lp:online-tif}),

\begin{align*}
p(i, x) + \sum_{k \ne i} \sum_{y \in \Support} f_k(y) \cdot
p(k, y)
&= \frac12 p_{ix}^* + \frac12 \sum_{k \ne i} \sum_{y \in
\Support} f_k(y) \cdot p_{ky}^*
\le \frac12 + \frac12 \sum_{k = 1}^n \sum_{y \in \Support}
f_k(y) \cdot p_{ky}^* \le 1
\end{align*}

%For the objective value,
%
%\macomment{TODO: THIS IS THE OTHER WAY AROUND}
%\[ (\text{\ref{lp:offline}}) = \sum_{i = 1}^n \sum_{x \in
%\Support} x \cdot p_{ix} \cdot f_i(x)
%= \frac12 \sum_{i = 1}^n \sum_{x \in \Support} x \cdot p(i, x) \cdot
%f_i(x) = \frac12 (\text{\ref{lp:online-tif}}) \]

\maedit{
Now we compute the
objective value of (\ref{lp:online-tif}),
}

\maedit{
\[
\sum_{i = 1}^n \sum_{x \in \Support} x \cdot f_i(x) \cdot
p(i, x)
= \frac12 \sum_{i = 1}^n \sum_{x \in \Support} x \cdot
f_i(x) \cdot p_{ix}^*
= \frac{C^*}{2}
\]
}

In fact, if we assume as before the special solution
$p_{ix}^* = p_x^*$ guaranteed by
\Cref{lem:offline-structural}, we can further prove that the
constraints of (\ref{lp:online-iif}) are satisfied for any
fixed arrival ordering $\pi \in S_n$, meaning that the
algorithm we designed in this proof will also be an IIF
stopping rule: $ p(x) + \sum_{k = 1}^{n-1} \sum_{y \in
\Support} f_{\pi(k)}(y) \cdot p(y) \le \frac12 p_x^* +
\frac12 \sum_{i = 1}^n \sum_{y \in \Support} f_i(y) \cdot
p_y^* \le 1$.

As for tightness of the inequality, the same
$\delta$-perturbation of the standard prophet inequality
instance presented in \Cref{thm:main1} suffices since for
permutation $\pi_1 = (1, 2)$, no algorithm (TIF or not) can
do better than $1 + \bigO(\eps)$ whereas a prophet gets at
least $2 + \bigO(\eps)$.
\end{proof}

\begin{theorem}[{\maedit{TIF Right Arc}}]
\label{thm:tif-right-arc}
Let $\I = (F_i)_{i = 1}^n$ be an instance of the
hiring problem. Let $\{\OptOn^{\pi}\}_{\pi \in S_n}$ be a
family of the optimal stopping rules for $\I$ (each being
optimal for the particular arrival ordering $\pi$).
There exists a family of TIF stopping rules
$\{\OnTIF^{\pi}\}_{\pi \in S_n}$ for the instance $\I$ such
that,

\[ \expect[ \OnTIF^{\pi} ] \ge \frac12 \cdot \expect[
\OptOn^{\pi} ],\quad \forall \pi \in S_n \]

Moreover, the inequality is tight.
\end{theorem}
\begin{proof}
Observe that the $\ge 1/2$ part is a consequence of
\Cref{thm:main2}. This is because no online algorithm can
perform better that an offline one in expectation, therefore
if an Online,TIF algorithm is $1/2$-competitive compared to
an Offline, the bound carries over when comparing to
Online algorithms which perform at most as good.

For the tightness, we consider again instance (\ref{eq:ex}).
We claim that no TIF stopping rule can perform better than
$1 - \bigO(\eps)$ and also there exists an online algorithm for
the same instance with arrival ordering $\pi_2 = (2, 1)$
which achieves performance $2 + \bigO(\eps)$. The later
claim is easy to see since without the
$\delta$-perturbation, inspecting $X_2$ first gives
advantage to the stopping rule to and performs as well as an
offline prophet. For the former
claim, one needs to solve the LP in (\ref{lp:online-tif})
for this specific instance to verify that the optimal
solution has objective value $1 - \bigO(\eps)$. We omit the
computation here.
\end{proof}

\section{An impossibility result}
\label{sec:impossibility}

We've seen how fairness (IIF, TIF or both) is achievable
in the online setting while guaranteeing a performance at
least half that of a prophet. This means that requiring
fairness does not hurt the expected performance of an online
decision maker (in the worst case over instances) more than
the loss induced purely by the online nature of the setting
when compared to an offline one. However, fairness comes at
a cost which is not immediately visible when comparing
expected performance.

Consider the online IIF
$1/2$-competitive algorithm
derived in \Cref{thm:main1}. By definition, $p(x) =
\frac12 p_x^*$ so the expected number of candidates hired by
the algorithm is $\sum_{i = 1}^n \sum_{x \in \Support} p(x)
\cdot f_i(x) = \frac12 \sum_{i = 1}^n \sum_{s \in \Support}
p_x^* \cdot f_i(x) \le \frac12$ using the constraint of
(\ref{lp:offline}). Since the stopping rules always hires at most
one person, this means the probability of hiring some
candidate is at most $1/2$. So there is at least $1/2$
probability of not hiring anyone!

Sometimes, not hiring anyone could be undesirable even when
on average a stopping rule is performing well.  It is
natural to ask, is there a fair (IIF/TIF) stopping rule with
non-zero competitive ratio when compared to a prophet which
hires some candidate with probability $1$?
A property like that, which requires that a stopping rule
always hires some candidate has also been studied in the
context of secretary problem by \cite{Buchbinder} where it
is called ``must-hire'' property.

In this section, we give a negative answer to this question
for both IIF and TIF stopping rules. However, our work
leaves open the question of whether there are \emph{offline}
algorithms in the  Offline, IIF or the Offline, TIF  setting
which satisfy the must-hire constraint.

\begin{theorem}
%There exists an instance $\I = ((\F_i)_{i = 1}^n, \pi)$ of
%the hiring problem such that
%there is no stopping rule $\ALG$ for $\I$ in the Online, IIF
%setting satisfying $\Pr[\ALG \text{ hires some candidate}] =
%1$ and $\expect[\ALG] \ge c \cdot \expect[\OptOff]$ for any
%$c > 0$.
For any $\eps > 0$, there exists an instance $\I_{\eps} =
((\F_i)_{i = 1}^n, \pi)$ of the hiring problem
such that for any IIF stopping rule $\ALG$ for $\I_{\eps}$
with $\Pr[\ALG \text{ hires exactly one candidate}] = 1$, we have:

\[ \expect[\ALG] < \eps \cdot \expect[ \max_i X_i ] \]
\end{theorem}
\begin{proof}
The constraint that a stopping rule always hires can be
expressed as a linear constraint in terms of $p(x)$
by letting

\[ q_i = \Pr[\ALG \text{ hires } i] = \sum_{x
\in \Support} \Pr[\ALG \text{ hires } i \mid X_i = x
]\cdot \Pr[X_i = x] = \sum_{x \in \Support} p(x) \cdot
f_i(x)
\]

Then we require that,

\[
\expect[\text{\# of
candidates hired}] = \sum_{i = 1}^n q_i = \sum_{i = 1}^n
\sum_{x \in \Support} p(x) \cdot f_i(x) = 1
\]

Append this constraint to (\ref{lp:online-iif}) to get an augmented LP.
Any solution to this augmented LP can be turned back into an
stopping rule which hires at most one person as we argued in
the previous sections and the extra constraint ensures that
the stopping rule derived hires exactly one person.
\maedit{Denoting by $\OptOnMHIIF$ the value picked by the optimal,
must-hire, IIF stopping rule, we thus have that the optimal
objective value of the augmented LP is exactly
$\expect[\OptOnMHIIF]$.}

Re-write the first constraint by using the second constraint
as follows:
\begin{align*}
\forall x \in \Support:
p(x) + \sum_{k = 1}^{n-1}
\sum_{y \in S} f_{\pi(k)}(y) \cdot p(y)\le 1
&\Leftrightarrow
p(x) + \underbrace{
\sum_{i = 1}^n \sum_{y \in \Support} f_{i}(y)
\cdot p(y)
}_{=1}
- \sum_{y \in \Support} f_{\pi(n)}(y) \cdot p(y)
\le 1\\
&\Leftrightarrow
p(x) \le \sum_{y \in \Support} f_{\pi(n)}(y) \cdot
p(y)
%,\quad \forall x \in \Support
\end{align*}

The last inequality says that all values $p(x)$ (which are
real numbers in $[0, 1]$) should be bounded above by a
convex combinations of the set of all $\{p(x)\}_{x \in
\Support}$. The only way for this to happen is if $p(x) =
p(y) = p$ for all $x, y \in \Support$.

Using now the extra constraint, we can compute $p$ as
follows:

\[
\sum_{i = 1}^n \sum_{x \in \Support} p \cdot f_i(x) = 1
\Rightarrow p \cdot \sum_{i = 1}^n 1 = 1 \Rightarrow p =
1/n
\]

Substituting back in the objective we get,

\[ \expect[\OptOnMHIIF] = \sum_{i =
1}^n \sum_{x \in \Support} x \cdot f_i(x) \cdot p =
\frac{1}{n} \sum_{i = 1}^n \expect[X_i]
\]

%This means that the
%optimal, must-hire IIF stopping rule has a
%performance equal to a $\frac1n$ of the expected values of
%$i$.
Therefore, we can take any instance such that $\sum_{i =
1}^n \expect[X_i]$ is very close to $\expect[\max_{i = 1}^n
X_i]$ and make $n$ large enough.

More precisely, given $\eps > 0$ define $\I_{\eps}$
to contain $n = \left\lceil \frac{2\ln 2}{\ln\left(
\frac{1}{1 - \eps/2} \right)} \right\rceil$
i.i.d.~random variables distributed as follows:
\[
\forall i \in [n]:
X_i = \begin{cases}
2/\eps, &\wp \eps/2\\
0,& \wp 1 - \eps/2
\end{cases}
\]

Then $\expect[\OptOnMHIIF] = 1$ and it can be shown that
$\expect[\max_i X_i] > 1/\eps$.
%and we get a ratio of $1/n$ which goes
%arbitrarily close to $0$ as $n \to \infty$.
\end{proof}

\begin{theorem}
%There exists an instance $\I = \{\F_i\}_{i = 1}^n$ of the
%hiring problem such that there is no stopping family of
%stopping rules $\{\ALG\}^{\pi}$ for $\I$ in the Online, TIF
%setting satisfying $\Pr[\ALG^{\pi} \text{ hires some
%candidate}] = 1$ and $\expect[\ALG^{\pi}] \ge c \cdot
%\expect[\OptOff]$ for any $c > 0$.
For any $\eps > 0$, there exists an instance $\I_{\eps} =
(\F_i)_{i = 1}^n$ of the hiring problem such that for any
TIF family of stopping rules $\{\ALG^{\pi}\}_{\pi \in S_n}$
with $\Pr[\ALG^{\pi} \text{ hires exactly one candidate}] =
1$ for all $\pi$, we have:

\[ \expect[\ALG^{\pi} ] \le \eps \cdot \expect[ \max_i X_i
], \forall \pi \in S_n \]
\end{theorem}

The proof uses similar tools as the proof of the previous
theorem and is omitted.

\section{Single-Sample and Double-Sample Fair Prophet Inequalities}
\label{sec:double-sample}

Suppose our algorithm doesn't know the distributions
$\F_1,\ldots,\F_n$ but is 
given independent samples from them. 
In this section we 
present a $\frac12$-competitive 
offline selection rule that 
satisfies IIF, given one 
sample from each distribution.
Then we present a
$\frac19$-competitive family of stopping rules
satisfying both IIF and TIF,
given {\em two} independent 
samples from each distribution.

In the following algorithms, the 
elements to be selected are denoted
$X_1,\ldots,X_n$; additional 
independent samples
from the same distributions 
are denoted as 
$Y_1,\ldots,Y_n$ 
and $Z_1,\ldots,Z_n$.
The algorithms that we 
analyze are comparison-based,
i.e.~their decisions are based
on comparing pairs of 
elements of the multiset 
$W=\{X_1,Y_1,Z_1,\ldots,X_n,Y_n,Z_n\}.$
It will be convenient to assume that
all such comparisons are strict, i.e.~that
no two elements of $W$ are equal.
To ensure strict comparisons, we 
assume that our algorithms sample
a uniformly random tie-breaking priority
$\psi(w)$ in $[0,1]$,
for each $w \in W$. Then when comparing
two elements $w,w' \in W,$ if $w$ and $w'$ 
are of equal value, the
outcome of the comparison is 
determined by comparing $\psi(w)$
with $\psi(w')$. (The probability
of $\psi(w)=\psi(w')$ is zero because
they are sampled from a distribution
with no point masses.)

% We assume the distributions
% $F_1,\ldots,F_n$ are continuous
%  (i.e., have no point masses) so
% that the $3n$ values 
% $X_1,Y_1,Z_1,\ldots,X_n,Y_n,Z_n$
% are all distinct, with probability 1.

\begin{algorithm}
\caption{Single-sample offline algorithm}\label{alg:single-sample}
\KwData{$X_i, Y_i \sim F_i$}
Let $i \in [n]$ be such 
that $X_i = \max \{ X_1,\ldots,X_n\}$.\\
\eIf{$X_i > Y_i$}{
{\bf Hire} $X_i$.\\
}{
{\bf Hire} no candidate.
}
\end{algorithm}

\begin{lemma}
\Cref{alg:single-sample} is IIF and $\frac12$-competitive.
\end{lemma}
\begin{proof}
For any $x > 0$ and $i \in [n]$ , 
\[
 \Pr[ \mbox{Hire } i \mid X_i=x ] =
 \Pr[ (\forall j \neq i: X_j < x ) \land
    (Y_i < x) ] =
 \prod_{j=1}^n F_j(x) .
\]
The right side does not depend on
$i$, as required by the definition
of IIF.

Let $M = 
\max \{X_1,Y_1,X_2,Y_2,\ldots,X_n,Y_n\}.$
When $M = X_i$ for some $i$,
\Cref{alg:single-sample} is 
guaranteed to hire $X_i$.
Therefore, if we define 
\[
    \zeta = \begin{cases}
        1, & \mbox{if } M \in \{X_1,\ldots,X_n\}\\
        0, & \mbox{if } M \in \{Y_1,\ldots,Y_n\}
        \end{cases}
\]
we have $\ALG \ge \zeta \cdot M$.
Observing that $\expect[ \zeta \mid M] = \frac12$ for all values of $M$, we find
that
\[
  \expect[ \ALG ] \ge
  \expect[ \zeta \cdot M ] =
  \expect[ \expect[\zeta \mid M] \cdot M ] =
  \frac12 \cdot \expect[ M ] \ge
  \frac12 \cdot \expect[ \max\{X_1,\ldots,X_n\} ] .
\]
\end{proof}

Next we present and analyze a $\frac19$-competitive
stopping rule that satisfies TIF and IIF given two
independent samples $Y_i,Z_i$ from each distribution.
The stopping rule is defined in \Cref{alg:double-sample}
below.

\begin{algorithm}
\caption{Double-sample online algorithm}\label{alg:double-sample}
\KwData{$Y_i, Z_i \sim F_i, \pi \in S_n$}
Let $Y_* = \max \{ Y_1,\ldots,Y_n \}.$\\
\For{$t = 1,2,\ldots,n$}{
Observe $X_{\pi(t)} \sim \F_{\pi(t)}$.\\
\If{$X_{\pi(t)} > Y_*$ and ($X_{\pi(s)} < Y_*$ for all $s <
t$) and ($Z_{\pi(s)} < Y_*$ for all $s \ge t$)}{
{\bf Hire} $X_{\pi(t)}$.
}
}
\end{algorithm}

First, observe that this is a well-defined stopping rule:
given the samples $Y_1,\ldots,Y_n,Z_1,\ldots,Z_n$,
the criterion for hiring $X_{\pi(t)}$ depends only on the 
values of $X_{\pi(1)},\ldots,X_{\pi(t)}$. Furthermore \Cref{alg:double-sample} 
never selects more than one element: if it hires $X_{\pi(t)}$
then $X_{\pi(t)} > Y_*$, which means that for all $t' > t$ the
condition for hiring $X_{\pi(t')}$ will not be satisfied.

\begin{lemma} \label{lem:double-sample-fair}
\Cref{alg:double-sample} satisfies TIF and IIF
\end{lemma}
\begin{proof}
For any $x, y \in \Support$, and any index $i$,
let us calculate the conditional probability
of hiring $i$ given that $X_i=x$ and $Y_* = y$.
If $x < y$ then this probability is zero.
Otherwise, suppose $i = \pi(t)$. 
We have
\begin{align*}
    \Pr[\mbox{Hire } i \mid X_i = x, Y_* = y] & =
    \Pr[X_{\pi(s)} < y \mbox{ for all } s < t] \cdot
    \Pr[Z_{\pi(s)} < y \mbox{ for all } s \ge t] \\ 
    & =
    \prod_{s=1}^n F_{\pi(s)}(y) = \prod_{j=1}^n F_j(y) . 
\end{align*}
Summing over the possible values of $y$,
\[
    \Pr[\mbox{Hire } i \mid X_i = x] =
    \sum_{y \in \Support}
    \Pr[Y_* = y] \cdot \Pr[\mbox{Hire } i \mid X_i = x, Y_*
    = y]
     = 
    \sum_{y \in \Support}  \Pr[Y_* = y] 
    \prod_{j=1}^n F_j(y) .
\]
The right side depends on neither $\pi$ nor $i$,
so the stopping rule is both TIF and IIF.
\end{proof}

\begin{lemma}
\Cref{alg:double-sample} is $\frac19$-competitive.
\end{lemma}
\begin{proof}
    Recall the multiset $W = 
    \{X_1,Y_1,Z_1,\ldots,X_n,Y_n,Z_n\}$,
    and recall that its elements are totally
    ordered using random priorities to break
    ties between elements of $W$ whose values
    are equal. Let $W_* > W_{**}$ denote 
    the two largest values in $W$. We will
    use $\mathcal{E}$ to
    denote the event that $W_* \in \{X_1,\ldots,X_n\}$
    and $W_{**} \in \{Y_1,\ldots,Y_n\}$.
    Conditional on $\mathcal{E}$, \Cref{alg:double-sample}
    is assured of hiring the largest element of $W$.
    Conditional on the contents of
    the set $W$ --- but not the partition of
    its elements into $X = \{X_1,\ldots,X_n\}, \,
    Y = \{Y_1,\ldots,Y_n\}, Z = \{Z_1,\ldots,Z_n\}$ ---
    the probability that $W_* \in X$ is $\frac13$
    and the conditional probability that $W_{**} \in Y$
    given $W_* \in X$ is at least $\frac13$. (It 
    is $\frac12$ if $W_{**}$ and $W_*$ are samples 
    from the same distribution, and $\frac13$ if
    they are from different distributions.) Hence,
    $\Pr[\mathcal{E} \mid W] \ge \frac19$ and 
    \[
        \expect [ \ALG ] = 
        \expect[ \expect[ \ALG \mid W ] ] \ge
        \expect[ W_* \cdot \Pr[\mbox{Hire } W_* | W] ] \ge
        \frac19 \cdot \expect[ W_* ] \ge 
        \frac19 \cdot \expect[ \max\{X_1,\ldots,X_n\} ] .
    \]
\end{proof}

\section*{Acknowledgements}

\maedit{
The authors would like to thank Kate Donahue for providing
insightful feedback during the early stages of the
development of this work.
}

\bibliographystyle{alpha}
\bibliography{bibliography}

\end{document}